\documentclass[a4paper,10pt]{amsart}

\usepackage{amssymb,epsfig,amsmath,amsthm}
\usepackage[font=small,format=plain,labelfont=bf,up,textfont=it,up]{caption}
\usepackage{graphicx,subfig}
\graphicspath{{./Figs/}}
\numberwithin{equation}{section}
\DeclareMathOperator*{\argmin}{arg\,min}

\title[A Penalty Method for Solving HJB Equations]{A Penalty Method for the Numerical Solution of Hamilton-Jacobi-Bellman (HJB) Equations in Finance}
\author{J. H. Witte}
\author{C. Reisinger}
\address{Mathematical Institute\\
University of Oxford}
\date{October 2010}
\email{\texttt{[\,witte\,,\,reisinge\,]\,@\,maths.ox.ac.uk}}
\thanks{This research was supported by Balliol College, University of Oxford, and the
UK Engineering and Physical Sciences Research Council (EPSRC)}

\begin{document}

\begin{abstract}
We present a simple and easy to implement method for the numerical
solution of a rather general class of Hamilton-Jacobi-Bellman (HJB) equations.
In many cases, classical finite difference discretisations can be shown to converge to the unique viscosity solutions of the considered problems.
However, especially when using fully implicit time stepping schemes with their desirable stability properties, one is still faced with the considerable task of solving the resulting nonlinear discrete system. In this paper, we introduce a penalty method which approximates the nonlinear discrete system to an order of $O(1/\rho)$, where $\rho>0$ is the penalty parameter, and we show that an iterative scheme can be used to solve the penalised discrete problem in finitely many steps. We include a number of examples from mathematical finance
for which the described approach yields a rigorous numerical scheme and present numerical results.
\bigskip

\hspace{-.45cm}\textit{Key Words:} HJB Equation, Numerical Solution, Penalty Method, Convergence Analysis, Viscosity Solution\bigskip

\hspace{-.45cm}\textit{2010 Mathematics Subject Classification:} 65M12, 93E20\bigskip

\end{abstract}

\maketitle

\newtheorem{theorem}{Theorem}[section]
\newtheorem{la}[theorem]{Lemma}
\newtheorem{cor}[theorem]{Corollary}
\newtheorem{remark}[theorem]{Remark}
\newtheorem{prob}[theorem]{Problem}	
\newtheorem{definition}[theorem]{Definition}
\newtheorem{alg}[theorem]{Algorithm}
\newtheorem{prop}[theorem]{Proposition}

\section*{Introduction}

In their most general form, Hamilton-Jacobi-Bellman (HJB) equations arise when applying Bellman's principle of optimality to problems of stochastic optimal control, in which case the solution to the HJB equation can often be found to solve the control problem (cf.\,\cite{XYZ_StochControlHJBBook}). There is a wide range of practical applications to optimal control, amongst others in mathematical finance, and many examples are listed in \cite{Fleming_Soner_ControlledMarkovProcesses,KushnerDupuis_StochasticControlProblems_ContTime}.
Naturally, the high practical relevance gives rise to the question of how to best solve stochastic control problems numerically. Early (and still relevant) methods rely on Markov chain approximation, of which an extensive overview can be found in \cite{Kushner_NumMethodsStochControl_ContTime, KushnerDupuis_StochasticControlProblems_ContTime, Fleming_Soner_ControlledMarkovProcesses}. The notion of viscosity solution, introduced in \cite{UsersGuide_ViscositySols}, together with the most helpful results of \cite{BarlesMainArticle}, permits another line of attack, namely the direct solution of the corresponding HJB equations by discretisation methods; for details, see \cite{BarlesMainArticle, Barles_ErrorBoundsMonotoneApproxSchemes, Barles_OntheConvergenceRate_ApproxHJBEq} and references therein.
It is interesting to note that both approaches can be found to be closely related, which is most apparent for explicit time stepping methods (e.g.\,cf.\,\cite{Song_MarkovChainHJBEq}).\bigskip

\begin{figure}[ht]
\centering
\includegraphics[width=10cm,height=6.5cm]{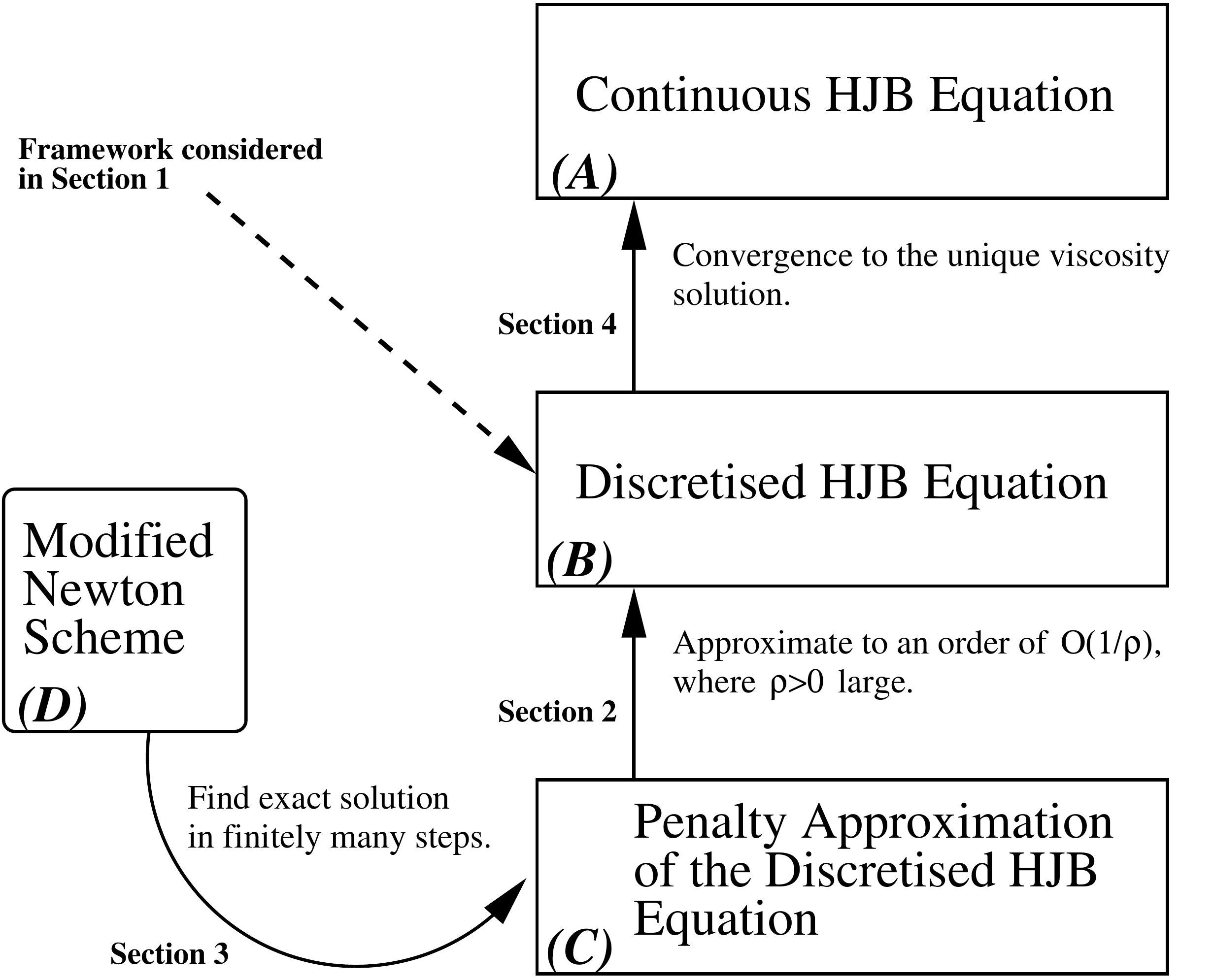}
\caption{To numerically approximate the HJB equation in (A), we consider a discrete version (B).
Given some conditions that have to be checked for every particular example, the grid convergence
of (B) to (A) can be guaranteed.
We then approximate (B)
by the penalty problem (C), which in turn can be solved exactly using algorithm (D).}
\label{fig:StructureOfPaper}
\end{figure}

In this paper, we are concerned with the numerical solution of a certain kind of HJB equation with a finite control set.
Using results of \cite{BarlesMainArticle}, where it was -- roughly speaking -- shown that every stable and consistent discretisation converges to the financially relevant solution, one can often find numerical approximations by employing standard finite difference schemes. However, the resulting discrete system is only straightforward to solve if fully explicit time stepping is used, in which case one suffers from undesirable stability constraints. When using fully implicit or weighted time stepping schemes, one faces a nonlinear discrete system, which is generally nontrivial to solve. Existing algorithms are, in one way or another, all related to policy iteration techniques, which have been used in the setting of Markov chain approximation \cite{KushnerDupuis_StochasticControlProblems_ContTime,Santos_PolicyIteration} as well as for HJB equations \cite{Forsyth_Controlled_HJB_PDEs_Finance, Forsyth_NonlinearPDEFinance, Wang_Forsyth_MaximalUSeCentralDifferences_HJBFinance}.
In the present paper, we introduce a novel approach by which the nonlinear discrete system can be solved using a penalisation technique. More precisely, we show how
a penalty approach that was devised for American option problems in \cite{ForsythQuadraticConvergence}
can be extended to solve discretised HJB equations provided the discretisation matrices obey a certain structure. (Penalty methods in general have a long history and have, for example, already been used extensively in the setting of variational inequalities in \cite{Lions_ApplVarIneqStochControl}.)
Examples from mathematical finance that fit into our framework include uncertain volatility models \cite{Avellaneda_UncertainVol},
transaction cost models \cite{Leland_TransactionCosts},
and unequal borrowing/lending rates and stock borrowing fees \cite{Amadori_NonlinINtegroDiffProb_OptionPricing_ViscositySolApproach,Bergman_DifferentialINterestRates}; a succinct overview of these and similar models can be found in \cite{Forsyth_Controlled_HJB_PDEs_Finance,Forsyth_NonlinearPDEFinance}.

\subsection*{Structure of this Paper} We aim to devise a penalty-based numerical scheme for the
solution of a rather general class of HJB equations; the broad concept is depicted in Figure \ref{fig:StructureOfPaper}.\medskip

\paragraph{\it{Section \ref{ProblemFormulation}}} We define the considered class of HJB
equations (denoted by \textit{(A)} in Fig.\,\ref{fig:StructureOfPaper}) and -- having fully implicit and weighted time stepping schemes in mind -- a fairly general
nonlinear discrete problem (denoted by \textit{(B)} in Fig.\,\ref{fig:StructureOfPaper}).
The structure of \textit{(B)} is chosen such that we have an unconditionally stable numerical scheme which converges
to \textit{(A)} if the latter has a unique viscosity solution and satisfies a strong comparison principle.\medskip

\paragraph{\it{Section \ref{Section_PenalisationDiscrProb}}}
The nonlinear discrete problem \textit{(B)} is nontrivial to solve and will be our main consideration. In a first step, we introduce a penalised
problem (denoted by \textit{(C)} in Fig.\,\ref{fig:StructureOfPaper}) which approximates \textit{(B)} to an order of $O(1/\rho)$, where $\rho>0$ is the penalty parameter.
It is worth pointing out that as the discretisation
in \textit{(B)} is stable (due to the monotonicity of the discretisation matrices), any errors introduced by the
penalisation will stay within bounds and, hence, will not compromise the overall numerical scheme.
Since \textit{(C)} is a nonlinear discrete problem itself, the feasibility of the penalty approach crucially depends on
the possibility of conveniently solving \textit{(C)} in a second step; in Section \ref{SectionModNewton}, we present
an algorithm doing exactly that.\medskip

\paragraph{\it{Section \ref{SectionModNewton}}} We present an iterative procedure (denoted by \textit{(D)} in Fig.\,\ref{fig:StructureOfPaper}) with
finite termination which finds the exact solution to \textit{(C)}.
Even though the theoretical bound on the number of iterations depends on the the number of grid points and the cardinality
of the finite control, the algorithm always converges -- independently of
the grid size and the penalty parameter -- in less than four steps for our numerical examples.\medskip

\paragraph{\it{Section \ref{ExHJBEqInFinance}}}
We give a detailed outline of how various examples from mathematical finance fit into our framework. Furthermore, for the HJB
equation arising in European option pricing with different borrowing/lending rates and stock borrowing fees, we show in detail
how convergence of \textit{(B)} to the unique viscosity solution of \textit{(A)} can be guaranteed.\medskip

\paragraph{\it{Section \ref{Section_NumResults}}} We conclude by presenting numerical results and compare our approach
to the method of policy iteration. Also, we provide a brief proof showing that policy iteration, just like the
iterative scheme proposed by us in Section \ref{SectionModNewton}, has finite termination when applied to the kind of HJB equation considered in this paper.

%
%
%
%
%
%
%
%
%
%

\section{Problem Formulation}\label{ProblemFormulation}

We begin by (slightly informally) introducing the type of HJB equation (see Problem\,\ref{DifferentialOperatorProblem}) considered
in this paper and for which we subsequently develop a fully coherent numerical scheme.
In Section \ref{ExHJBEqInFinance}, we will then show that the problems from mathematical finance
listed in the introduction all match our framework and that rigorous problem formulations can be obtained by employing the theory of viscosity solutions (as introduced in \cite{UsersGuide_ViscositySols}).

\begin{prob}\label{DifferentialOperatorProblem}
Let $\mathbb{S}$, $|\mathbb{S}|<\infty$, be a finite set of controls. Let $\Omega\subset\mathbb{R}^n$ (with $n\in\mathbb{N}$) be an open set.
Let $\mathcal{L}_s$\,, $s\in\mathbb{S}$, be a number of linear differential operators on $\Omega$.
Find a function $V:\Omega\to\mathbb{R}$ such that
\begin{equation}
\min\{\mathcal{L}_sV : s\in\mathbb{S}\}=0.\label{DifferentialOperatorProblem_Eq1}
\end{equation}
\end{prob}

At this stage, we are not yet concerned with possible boundary/initial conditions needed for the uniqueness of a solution to Problem \ref{DifferentialOperatorProblem} since formulation \eqref{DifferentialOperatorProblem_Eq1} is primarily motivational.
Nevertheless, we point out that (if we have an additional linear differential operator $\mathcal{L}^*$) formulation \eqref{DifferentialOperatorProblem_Eq1} is sufficiently general to include problems
of the form
\begin{equation}
\mathcal{L}^*V +\min\{\tilde{\mathcal{L}}_sV : s\in\mathbb{S}\}=0\label{DifferentialOperatorProblem_Reformulated}
\end{equation}
when defining
$\mathcal{L}_s := \tilde{\mathcal{L}}_s + \mathcal{L}^*$, $s\in\mathbb{S}$.
In many applications, $\mathcal{L}^*V$ in \eqref{DifferentialOperatorProblem_Reformulated} is taken to be the time derivative of $V$ (i.e. $\mathcal{L}^*V=\partial V/\partial t$). Also, the approach described in this paper works analogously when replacing the ``min'' in equations \eqref{DifferentialOperatorProblem_Eq1} and \eqref{DifferentialOperatorProblem_Reformulated} by ``max'', which will be explained in Remark \ref{SubAlsoOK}.

\begin{remark}\label{RelationalOperators_Remark}
Let $N\in\mathbb{N}$ and define $\mathcal{N}:=\{1,\ldots,N\}$.
Consider two vectors $x$, $y\in\mathbb{R}^N$ and a matrix $A\in\mathbb{R}^{N\times N}$. For any $i\in\mathcal{N}$, we denote by $(x)_i$ and $(A)_i$ the $i$-th coordinate and the $i$-th row of vector $x$ and matrix $A$, respectively. When writing $x\geq 0$, we generally mean that $(x)_i\geq 0$ for all $i\in\mathcal{N}$, and we take $z:=\min\{x,y\}$ to be the vector satisfying $(z)_i=\min\{(x)_i\,,(y)_i\}$ for all $i\in\mathcal{N}$.
The definitions extend trivially to other relational operators and to the maximum of two vectors.
\end{remark}

In the following, whenever the choice of $N\in\mathbb{N}$ is clear from the context, we will use the set $\mathcal{N}$ as introduced in Remark \ref{RelationalOperators_Remark} without explicitly saying so.\bigskip

We make the assumption that we can find a sensible (i.e. stable and convergent) discretisation of Problem \ref{DifferentialOperatorProblem} that results in discrete systems of the following form.

\begin{prob}\label{DiscreteProbDef}
Find $x\in\mathbb{R}^N$ such that
\begin{equation}
\min\{A_sx-b_s : s\in\mathbb{S}\}=0,\label{DiscreteProbDef_Eq1}
\end{equation}
where
\begin{itemize}
\item $b_s\in\mathbb{R}^N$, $s\in\mathbb{S}$, are vectors,
\item $A_s\in\mathbb{R}^{N\times N}$, $s\in\mathbb{S}$, are matrices with non-positive off-diagonal entries, positive diagonal entries, non-negative row sums, and at least one positive row sum,
\item and the positive row sum is equally located in all matrices $A_s$\,, $s\in\mathbb{S}$.
\end{itemize}
\end{prob}

In Section \ref{StableConvDiscretisations}, we will show that stable and convergent discretisations satisfying the setup of Problem \ref{DiscreteProbDef} can be found in many relevant cases (and, most notably, for many examples from finance), justifying our assumptions.
The conditions on the matrices are motivated by the fact that a square matrix with non-positive off-diagonal entries, positive diagonal entries, non-negative row sums, and at least one positive row sum is an M-matrix (cf.\,\cite{FiedlerSpecialMatrices,Varga_IterativeMatrixAnalysis}), which means, in particular, that it has a non-negative inverse; we shall make frequent use of this fact.

\begin{remark}\label{SubAlsoOK}
If the ``min'' in equation \eqref{DifferentialOperatorProblem} is replaced by ``max'', equation \eqref{DiscreteProbDef_Eq1} becomes
\begin{equation}
\min\{\tilde{A}_sx-\tilde{b}_s : s\in\mathbb{S}\}=0,\label{DiscreteProbDef_Eq1_Modified}
\end{equation}
with $\tilde{A}_s:=-A_s$ and $\tilde{b}_s:=-b_s$\,. In this case, $-\tilde{A}_s$\,, $s\in\mathbb{S}$, is an M-matrix, which means that $\tilde{A}_s$ has non-positive inverse $\tilde{A}^{-1}\leq 0$; therefore, all proofs that follow (Theorem \ref{DiscreteProbDef_UniqSol} and Sections \ref{Section_PenalisationDiscrProb} and \ref{SectionModNewton}) can be shown to hold for both situations as soon as they hold for one.
However, for the sake of simplicity, we solely focus on situation \eqref{DiscreteProbDef_Eq1}.
\end{remark}

\begin{theorem}\label{DiscreteProbDef_UniqSol}
There exists a unique solution to Problem \ref{DiscreteProbDef}.
\end{theorem}
\begin{proof}
The proof of existence will be given in Corollary \ref{PenaltyConvergenceToTrueSol}.
Suppose we have two solutions $x_1$ and $x_2$\,.
For every $i\in\mathcal{N}$, there exists an $s_i\in\mathbb{S}$ such that
\begin{equation*}
(A_{s_i})_i\,x_2-(b_{s_i})_i=0,
\end{equation*}
and it also is
\begin{equation*}
(A_{s_i})_i\,x_1-(b_{s_i})_i\geq0.
\end{equation*}
Denote by $A^*\in\mathbb{R}^{N\times N}$ the matrix consisting of the rows $(A_{s_i})_i$\,, $i\in\mathcal{N}$, which is an M-matrix by our assumptions on the discretisation (cf.\,Problem\,\ref{DiscreteProbDef}). We have
\begin{equation*}
A^*(x_1-x_2)\geq 0,
\end{equation*}
and $x_1-x_2\geq 0$ follows since $(A^*)^{-1}\geq 0$; swapping the arguments, we get $x_2-x_1\geq 0$, which then proves the theorem.
\end{proof}

\section{Penalisation of the Discrete Problem}\label{Section_PenalisationDiscrProb}

In this section, we discuss how the nonlinear discrete system in Problem \ref{DiscreteProbDef} can be approximated by another (also nonlinear) problem using a penalisation technique, and it will turn out that the penalised formulation can very efficiently be solved numerically.
Our penalty approximation to the discrete HJB equation can be seen as an extension of the penalty scheme for American options as introduced in \cite{ForsythQuadraticConvergence}.

\subsection{The Penalised Problem} Our penalty approximation looks as follows.

\begin{prob}\label{PenDiscreteProbDef}
Let $s_0\in\mathbb{S}$, $\rho>0$. Find $x_{\rho}\in\mathbb{R}^N$ such that
\begin{equation}
(A_{s_0}\,x_{\rho} - b_{s_0}) -\rho \sum_{s\in\mathbb{S}\backslash\{s_0\}}\max\{b_s-A_s\,x_{\rho}\,,0\} =0.\label{PenDiscreteProbDef_Eq1}
\end{equation}
\end{prob}

In this, $\rho>0$ is called the penalty parameter, and \eqref{PenDiscreteProbDef_Eq1} can be thought of as applying a penalty of intensity $\rho$ whenever an expression $A_s-b_s\,x_{\rho}$\,, $s\in\mathbb{S}\backslash\{s_0\}$, turns negative; by pushing up the expression $\rho \sum_{s\in\mathbb{S}\backslash\{s_0\}}\max\{b_s-A_s\,x_{\rho}\,,0\}$, the penalty parameter causes $(A_{s_0}\,x_{\rho} - b_{s_0})$ and, implicitly, $x_{\rho}$ to grow until no constraints are violated anymore.

\begin{theorem}\label{PenDiscreteProbDef_Uniq}
There exists a unique solution to Problem \ref{PenDiscreteProbDef}.
\end{theorem}
\begin{proof}
The proof of existence will be given in Theorem \ref{NewtonLimitSolvesPenProblem}.
Suppose we have two solutions $x_{\rho,1}$ and $x_{\rho,2}$\,.
It is
\begin{equation}
(A_{s_0}\,x_{\rho,1} - b_{s_0}) -\rho \sum_{s\in\mathbb{S}\backslash\{s_0\}}\max\{b_s-A_s\,x_{\rho,1}\,,0\} =0\label{PenDiscreteProbDef_Uniq_Eq1}
\end{equation}
and
\begin{equation}
(A_{s_0}\,x_{\rho,2} - b_{s_0}) -\rho \sum_{s\in\mathbb{S}\backslash\{s_0\}}\max\{b_s-A_s\,x_{\rho,2}\,,0\} =0,\label{PenDiscreteProbDef_Uniq_Eq2}
\end{equation}
and, subtracting \eqref{PenDiscreteProbDef_Uniq_Eq2} from \eqref{PenDiscreteProbDef_Uniq_Eq1}, we get
\begin{multline}
A_{s_0}(x_{\rho,1} - x_{\rho,2}) -\rho \sum_{s\in\mathbb{S}\backslash\{s_0\}}\max\{b_s-A_s\,x_{\rho,1}\,,0\} \\
+\rho \sum_{s\in\mathbb{S}\backslash\{s_0\}}\max\{b_s-A_s\,x_{\rho,2}\,,0\}=0.\label{PenDiscreteProbDef_Uniq_Eq2.8}
\end{multline}
Now, let $i\in\mathcal{N}$ and $r\in\mathbb{S}\backslash\{s_0\}$, and consider
\begin{multline}
 \big(-\rho\max\{b_r-A_r\,x_{\rho,1}\,,0\} + \rho\max\{b_r-A_r\,x_{\rho,2}\,,0\}\big)_i\\
 = -\rho\max\{(b_r)_i-(A_r\,x_{\rho,1})_i\,,0\} + \rho\max\{(b_r)_i-(A_r\,x_{\rho,2})_i\,,0\},\label{PenDiscreteProbDef_Uniq_Eq3}
\end{multline}
where we distinguish between the following different cases.
\begin{itemize}
\item[(i)] If $\max\{(b_r)_i-(A_r\,x_{\rho,1})_i\,,0\} = \max\{(b_r)_i-(A_r\,x_{\rho,2})_i\,,0\}=0$, then expression \eqref{PenDiscreteProbDef_Uniq_Eq3}
is identical to zero.
\item[(ii)] If $\max\{(b_r)_i-(A_r\,x_{\rho,1})_i\,,0\}>0$ and $\max\{(b_r)_i-(A_r\,x_{\rho,2})_i\,,0\}>0$, then expression 
\eqref{PenDiscreteProbDef_Uniq_Eq3} can be written as
$\rho(A_r\,x_{\rho,1})_i -\rho (A_r\,x_{\rho,2})_i = \big(\rho A_r(x_{\rho,1} -x_{\rho,2})\big)_i$\,.
\item[(iii)] If $\max\{(b_r)_i-(A_r\,x_{\rho,1})_i\,,0\}>0$ and $\max\{(b_r)_i-(A_r\,x_{\rho,2})_i\,,0\}=0$, then expression
\eqref{PenDiscreteProbDef_Uniq_Eq3} is negative, i.e.
$-\rho(b_r)_i + \rho(A_r\,x_{\rho,1})_i  < 0$.
\item[(iv)] If $\max\{(b_r)_i-(A_r\,x_{\rho,1})_i\,,0\}=0$ and $\max\{(b_r)_i-(A_r\,x_{\rho,2})_i\,,0\}>0$, then expression
\eqref{PenDiscreteProbDef_Uniq_Eq3} equals
\begin{multline*}
\rho(b_r)_i - \rho(A_r\,x_{\rho,2})_i\\
= \rho(b_r)_i - \rho(A_r\,x_{\rho,2})_i
 + \rho(A_r\,x_{\rho,1})_i - \rho(b_r)_i - \big(\rho(A_r\,x_{\rho,1})_i - \rho(b_r)_i\big)\\
=  \rho \big(A_r(x_{\rho,1}-x_{\rho,2})\big)_i + \rho(b_r)_i - \rho(A_r\,x_{\rho,1})_i\,,
\end{multline*}
in which $\rho(b_r)_i - \rho(A_rx_{\rho,1})_i\leq 0$.
\end{itemize}

Now, noting that the M-matrix property of $A_{s_0}$ is preserved if we add rows of
the matrices $A_s$\,, $s\in\mathbb{S}\backslash\{s_0\}$, to some of its rows, expression \eqref{PenDiscreteProbDef_Uniq_Eq2.8} implies, in the light of the four cases discussed above, that
there exists an M-matrix $A^*\in\mathbb{R}^{N\times N}$ satisfying
\begin{equation*}
A^*(x_{\rho,1} - x_{\rho,2}) \geq 0,
\end{equation*}
which then gives $(x_{\rho,1} - x_{\rho,2}) \geq 0$ since $(A^*)^{-1}\geq 0$; by the same reasoning, we can also get $(x_{\rho,2} - x_{\rho,1}) \geq 0$, which then completes the proof.
\end{proof}

\subsection{Properties of the Penalised Problem}

What follows are properties of the penalised problem. We prove that
the penalty approximation converges to the original problem as $\rho\to\infty$ and provide an estimate for the penalisation error.

\begin{la}\label{PenDiscreteProbDef_BoundedIndepRho}
Let $x_{\rho}$ be the solution to Problem \ref{PenDiscreteProbDef}.
There exists a constant $C>0$, independent of $\rho$, such that $\|x_{\rho}\|_{\infty}\leq C$.
\end{la}
\begin{proof}
We rewrite equation \eqref{PenDiscreteProbDef_Eq1} to get
\begin{equation}
A_{s_0}\,x_{\rho} =\rho \sum_{s\in\mathbb{S}\backslash\{s_0\}}\max\{b_s-A_s\,x_{\rho}\,,0\} + b_{s_0}.\label{PenDiscreteProbDef_BoundedIndepRho_Eq1}
\end{equation}
Hence, for every component $i\in\mathcal{N}$, we have that
\begin{equation*}
(A_{s_0}\,x_{\rho})_i= (b_{s_0})_i\quad\text{or}\quad\exists\ s\in\mathbb{S}\backslash\{s_0\}\text{ s.t. }(A_s\,x_{\rho})_i < (b_s)_i\,.
\end{equation*}
Therefore, we may, for $i\in\mathcal{N}$, define $s_i\in\mathbb{S}$ to be such that $(A_{s_i}\,x_{\rho})_i\leq (b_{s_i})_i$
is satisfied, and, introducing $A^*\in\mathbb{R}^{N\times N}$ to be the matrix having as $i$-th row the $i$-th row of $A_{s_i}$\,, $i\in\mathcal{N}$, and defining $b^*\in\mathbb{R}^N$ correspondingly, it is
\begin{equation*}
A^*\,x_{\rho}\leq b^*.
\end{equation*}
Since $A^*$ is an M-matrix, it is $(A^*)^{-1}\geq 0$, and we get $x_{\rho}\leq (A^*)^{-1}b^*$,
in which $(A^*)^{-1}$ and $b^*$ can be bounded independently of $\rho$ because there are only finitely many compositions that can be assumed.
To see that the negative part of $x_{\rho}$ is bounded as well, we note that from
\eqref{PenDiscreteProbDef_BoundedIndepRho_Eq1} we have
\begin{equation*}
x_{\rho} \geq -(A_{s_0})^{-1}b_{s_0}\,,
\end{equation*}
in which $(A_{s_0})^{-1}\geq 0$.
\end{proof}

\begin{la}\label{PenaltyApproxAccuracy_Theorem}
Let $x_{\rho}$ be the solution to Problem \ref{PenDiscreteProbDef}. There exists a constant $C>0$, independent of $\rho$, such that
\begin{align}
A_{s_0}\,x_{\rho}-b_{s_0}\geq&\ 0,\label{PenaltyApproxAccuracy_Theorem_Eq0.1}\\
A_s\,x_{\rho}-b_s\geq& - \frac{C}{\rho},\quad s\in\mathbb{S}\backslash\{s_0\},\label{PenaltyApproxAccuracy_Theorem_Eq0.2}
\end{align}
and, for every $i\in\mathcal{N}$, it is
\begin{equation}
(A_{s_0}\,x_{\rho}-b_{s_0})_i=0\label{PenaltyApproxAccuracy_Theorem_Eq0.3}
\end{equation}
or $\exists\ s\in\mathbb{S}\backslash\{s_0\}$ such that
\begin{equation}
(A_s\,x_{\rho}-b_s)_i \leq \frac{C}{\rho}.\label{PenaltyApproxAccuracy_Theorem_Eq0.4}
\end{equation}
Combined, properties \eqref{PenaltyApproxAccuracy_Theorem_Eq0.1}--\eqref{PenaltyApproxAccuracy_Theorem_Eq0.4} imply that $\|\min\{A_s\,x_{\rho}-b_s : s\in\mathbb{S}\}\|_{\infty}\leq C/\rho$.
\end{la}
\begin{proof}
For $s\in\mathbb{S}\backslash\{s_0\}$, it is
\begin{align*}
\rho(b_s - A_s\,x_{\rho})\leq&\ \rho \sum_{r\in\mathbb{S}\backslash\{s_0\}}\max\{b_r-A_r\,x_{\rho}\,,0\}\\
=&\ A_{s_0}\,x_{\rho} - b_{s_0}\,,
\end{align*}
and, applying Lemma \ref{PenDiscreteProbDef_BoundedIndepRho} to the last expression, we get that
\begin{equation}
A_s\, x_{\rho}-b_s \geq -\frac{C}{\rho}\label{PenaltyApproxAccuracy_Theorem_Eq1}
\end{equation}
for some constant $C>0$ independent of $\rho$.
Furthermore, we have that, for every component $i\in\mathcal{N}$, it is
\begin{equation*}
(A_{s_0}\,x_{\rho}-b_{s_0})_i=0
\end{equation*}
or $\exists\ s\in\mathbb{S}\backslash\{s_0\}$ such that
\begin{equation*}
(b_s-A_s\,x_{\rho})_i \geq 0,
\end{equation*}
which, together with \eqref{PenaltyApproxAccuracy_Theorem_Eq1}, extends to
\begin{equation*}
|(A_sx-b_s)_i| \leq \frac{C}{\rho}.
\end{equation*}
\end{proof}

\begin{cor}\label{PenaltyConvergenceToTrueSol}
Let $x_{\rho}$ be the solution to Problem \ref{PenDiscreteProbDef}.
As $\rho\to\infty$, $x_{\rho}$ converges to a limit $x^*$ which solves Problem \ref{DiscreteProbDef}.
\end{cor}
\begin{proof}
Since $(x_{\rho})_{\rho>0}$ is bounded (as seen in Lemma \ref{PenDiscreteProbDef_BoundedIndepRho}), it has a convergent subsequence, which we will not distinguish notationally; we denote the limit of this subsequence by $x^*$.
From \eqref{PenaltyApproxAccuracy_Theorem_Eq0.1} and \eqref{PenaltyApproxAccuracy_Theorem_Eq0.2}, we may deduce that 
\begin{align*}
A_s\,x^*-b_s\geq 0,\quad s\in\mathbb{S}.
\end{align*}
Now, consider a component $i\in\mathcal{N}$.
From \eqref{PenaltyApproxAccuracy_Theorem_Eq0.3} and \eqref{PenaltyApproxAccuracy_Theorem_Eq0.4}, we know that, for every $\rho>0$, there exists an $s_{i,\rho}\in\mathbb{S}$ such that
\begin{equation*}
|(A_{s_{i,\rho}}\,x_{\rho}-b_{s_{i,\rho}})_i|\leq \frac{C}{\rho}.
\end{equation*}
Recalling that $|\mathbb{S}|<\infty$, we can then infer that there exists once more a subsequence of $(x_{\rho})_{\rho> 0}$\,,
again not notationally distinguished, and an $s^*_i\in\mathbb{S}$ such that
\begin{equation}
|(A_{s^*_i}\,x_{\rho}-b_{s^*_i})_i|\leq \frac{C}{\rho}\quad\forall\ \rho>0,\label{PenaltyConvergenceToTrueSol_Eq1}
\end{equation}
which means
\begin{equation}
|(A_{s^*_i}\,x^*-b_{s^*_i})_i|=0.\label{PenaltyConvergenceToTrueSol_Eq2}
\end{equation}
Altogether, we see that $x^*$ solves Problem \ref{DiscreteProbDef}.
Since Problem \ref{DiscreteProbDef} has a unique solution (cf.\,Theorem\,\ref{DiscreteProbDef_UniqSol}), the just given result holds not only for subsequences of $(x_{\rho})_{\rho>0}$\,, but the whole sequence converges.
\end{proof}

\begin{theorem}\label{PenaltyConvergenceToTrueSol_ErrorEstimate}
Let $x_{\rho}$ and $x^*$ be the solutions to Problems \ref{PenDiscreteProbDef} and \ref{DiscreteProbDef}, respectively.
There exists a constant $C>0$, independent of $\rho$, such that
\begin{equation*}
\|x^*-x_{\rho}\|_{\infty}\leq\frac{C}{\rho}.
\end{equation*}
\end{theorem}
\begin{proof}
We continue in the setup of the proof of Corollary \ref{PenaltyConvergenceToTrueSol}.
From \eqref{PenaltyConvergenceToTrueSol_Eq1} and \eqref{PenaltyConvergenceToTrueSol_Eq2}, we know that
the original sequence $(x_{\rho})_{\rho>0}$ of solutions to Problem \ref{PenDiscreteProbDef} has a subsequence
$(x_{\rho'})_{\rho'>0}$ such that, for every component $i\in\mathcal{N}$, there exists an $s^*_i\in\mathbb{S}$ satisfying
\begin{equation*}
|(A_{s^*_i})_i\,(x_{\rho'}-x^*)_i|\leq \frac{C}{\rho'}\quad\forall\ \rho'>0.
\end{equation*}
Denote the matrix consisting of the rows $(A_{s^*_i})_i$\,, $i\in\mathcal{N}$, by $A^*\in\mathbb{R}^{N\times N}$.
Noting that $A^*$ is an M-matrix and hence invertible,
we obtain that
\begin{equation}
|x^*-x_{\rho'}|\leq \frac{\|(A^*)^{-1}\|_{\infty}\,C}{\rho'}\quad\forall\ \rho'>0,\label{PenaltyConvergenceToTrueSol_ErrorEstimate_Eq1}
\end{equation}
in which $\|(A^*)^{-1}\|_{\infty}$ can be bounded independently of $\rho$ because there are only finitely many compositions it can assume. Since, by the same arguments, every subsequence of the original sequence $(x_{\rho})_{\rho>0}$ can be shown
to have a subsequence satisfying \eqref{PenaltyConvergenceToTrueSol_ErrorEstimate_Eq1}, we deduce
that the whole sequence satisfies
\begin{equation*}
|x^*-x_{\rho}|\leq\frac{\hat{C}}{\rho}\quad\forall\ \rho>0
\end{equation*}
for some constant $\hat{C}>0$ independent of $\rho$.
\end{proof}

\section{An Iterative Method for Solving the Penalised Problem}\label{SectionModNewton}

Having seen in the previous section how penalisation can be used to approximate Problem \ref{DiscreteProbDef}, it remains to answer the question how to numerically solve the nonlinear discrete system \eqref{PenDiscreteProbDef_Eq1} appearing in the penalised formulation. To this end, we present an iterative method -- distantly resembling Newton's method -- 
which has finite termination and converges to the solution to Problem \ref{PenDiscreteProbDef}.
In the situation of American options, a similar algorithm was used in \cite{ForsythQuadraticConvergence}.\bigskip

We are trying to find $y\in\mathbb{R}^N$ which satisfies \eqref{PenDiscreteProbDef_Eq1}; hence, defining
\begin{equation*}
G(y):=(A_{s_0}\,y - b_{s_0}) -\rho \sum_{s\in\mathbb{S}\backslash\{s_0\}}\max\{b_s-A_s\,y,0\},\quad y\in\mathbb{R}^N,
\end{equation*}
we need to solve $G(y)=0$.

\begin{remark}
Let $\xi$ be a vector (or a matrix) in $\mathbb{R}^N$ (or $\mathbb{R}^{N\times N}$). For $y\in\mathbb{R}^N$, we denote by $\xi^y_s$\,, $s\in\mathbb{S}\backslash\{s_0\}$, the vector (or matrix) consisting 
of
\begin{itemize}
\item all rows of $\xi$ where the corresponding rows of $b_s-A_s\,y$ are positive
\item and where all other rows are set to zero.
\end{itemize} 
\end{remark}

\begin{alg}\label{ModifiedNewtonAlg}
Define
\begin{equation*}
J_G(y):=A_{s_0} +  \rho\sum_{s\in\mathbb{S}\backslash\{s_0\}}A^y_s\,,\quad y\in\mathbb{R}^N.
\end{equation*}
Let $x^0\in\mathbb{R}^N$ be some starting value. Then, for known $x^n$, $n\geq 0$, find $x^{n+1}$ such that
\begin{equation}
J_G(x^n)(x^{n+1}-x^n)= -G(x^n).\label{ModifiedNewtonAlg_Eq1}
\end{equation}
\end{alg}
Clearly, the just introduced algorithm bears strong resemblance to Newton's method. However, since (as we will see) it converges globally in finitely many steps, it is more powerful than standard Newton iteration, for which convergence can only be guaranteed in a sufficiently small neighbourhood of the solution.
Furthermore, the function $G(\cdot)$ is not differentiable, which means standard Newton iteration is not directly applicable. For general results on extensions of Newton methods to semi-smooth Newton methods
for the solution of particular non-smooth equations, see
\cite{ItoKunisch_SemiSmoothNewton_VarIneqFirstKind,ItoKunisch_SemiSmoothNewtonAndGlobalization} and references therein.\bigskip

The algorithm is well defined since $J_G$ is a non-singular matrix. Also, equation \eqref{ModifiedNewtonAlg_Eq1} can be written as
\begin{align*}
\big(A_{s_0} + \rho\sum_{s\in\mathbb{S}\backslash\{s_0\}}A^{x^n}_s\big)(x^{n+1}-x^n)
= - (A_{s_0}x^n - b_{s_0}) +\rho \sum_{s\in\mathbb{S}\backslash\{s_0\}}(b^{x^n}_s-A^{x^n}_sx^n),
\end{align*}
which is equivalent to
\begin{align}
\big(A_{s_0} + \rho\sum_{s\in\mathbb{S}\backslash\{s_0\}}A^{x^n}_s\big)x^{n+1}
= b_{s_0} +\rho \sum_{s\in\mathbb{S}\backslash\{s_0\}} b^{x^n}_s.\label{ModifiedNewtonAlg_Eq2}
\end{align}

\begin{la}\label{NewtonAlgMonotonicity}
Let $x^0\in\mathbb{R}^N$ be some starting value, and let $(x^n)^{\infty}_{n=0}$ be the sequence generated
by Algorithm \ref{ModifiedNewtonAlg}. It is $x^n\leq x^{n+1}$ for $n\geq 1$.
\end{la}
\begin{proof}
Writing $\eqref{ModifiedNewtonAlg_Eq2}$ for $x^n$ and $x^{n+1}$, we obtain
\begin{align*}
\big(A_{s_0} + \rho\sum_{s\in\mathbb{S}\backslash\{s_0\}}A^{x^n}_s\big)x^{n+1}
=&\ b_{s_0} +\rho \sum_{s\in\mathbb{S}\backslash\{s_0\}} b^{x^n}_s\\
\text{and}\quad
\big(A_{s_0} + \rho\sum_{s\in\mathbb{S}\backslash\{s_0\}}A^{x^{n-1}}_s\big)x^n
=&\ b_{s_0} +\rho \sum_{s\in\mathbb{S}\backslash\{s_0\}} b^{x^{n-1}}_s,
\end{align*}
and subtracting yields
\begin{multline}
\big(A_{s_0} + \rho\sum_{s\in\mathbb{S}\backslash\{s_0\}}A^{x^n}_s\big)(x^{n+1}-x^n)\\
=
\rho \sum_{s\in\mathbb{S}\backslash\{s_0\}} b^{x^n}_s
- \rho\sum_{s\in\mathbb{S}\backslash\{s_0\}}A^{x^n}_s x^n
-\rho \sum_{s\in\mathbb{S}\backslash\{s_0\}} b^{x^{n-1}}_s
+ \rho\sum_{s\in\mathbb{S}\backslash\{s_0\}}A^{x^{n-1}}_s x^n\,.\label{NewtonAlgMonotonicity_Eq1}
\end{multline}
In this, $A_{s_0} + \rho\sum_{s\in\mathbb{S}\backslash\{s_0\}}A^{x^n}_s$ is an M-matrix (and has a non-negative inverse), and, therefore, the proof is complete if we can show that the right hand side of expression \eqref{NewtonAlgMonotonicity_Eq1} is non-negative; to do this, we consider the rows $i\in\mathcal{N}$ of the right hand side of \eqref{NewtonAlgMonotonicity_Eq1} separately.
Let $s\in\mathbb{S}\backslash\{s_0\}$.
\begin{itemize}
 \item[(i)] If $(A^{x^n}_s)_i = 0$ and $(A^{x^{n-1}}_s)_i= 0$, it is
$\rho (b^{x^n}_s)_i
- \rho (A^{x^n}_s x^n)_i
-\rho (b^{x^{n-1}}_s)_i
+ \rho (A^{x^{n-1}}_s x^n)_i= 0$.
\item[(ii)] If $ (A^{x^n}_s)_i > 0$ and $(A^{x^{n-1}}_s)_i > 0$, it is
$ (A^{x^n}_s)_i = (A^{x^{n-1}}_s)_i$ and $(b^{x^n}_s)_i = (b^{x^{n-1}}_s)_i$\,, and, hence, it is $\rho (b^{x^n}_s)_i
- \rho (A^{x^n}_s x^n)_i
-\rho (b^{x^{n-1}}_s)_i
+ \rho (A^{x^{n-1}}_s x^n)_i= 0$.
\item[(iii)] If $ (A^{x^n}_s)_i > 0$ and $(A^{x^{n-1}}_s)_i = 0$, it is
$\rho (b^{x^n}_s)_i
- \rho (A^{x^n}_s x^n)_i
-\rho (b^{x^{n-1}}_s)_i
+ \rho (A^{x^{n-1}}_s x^n)_i
= \rho (b^{x^n}_s)_i
- \rho (A^{x^n}_s x^n)_i\geq 0.$
\item[(iv)]  If $(A^{x^n}_s)_i = 0$ and $(A^{x^{n-1}}_s)_i > 0$, it is
$\rho (b^{x^n}_s)_i
- \rho (A^{x^n}_s x^n)_i
-\rho (b^{x^{n-1}}_s)_i
+ \rho (A^{x^{n-1}}_s x^n)_i
= -\rho (b^{x^{n-1}}_s)_i
+ \rho (A^{x^{n-1}}_s x^n)_i\geq 0$, where the last inequality holds because $\max\{(b_s)_i-(A_s)_i\,x^n\,,0\}=0$.
\end{itemize}
This completes the proof.
\end{proof}

\begin{la}\label{NewtonAlg_FiniteTermination}
There exists a number $\kappa\in\mathbb{N}$ such that, for every starting value $x^0$\,, Algorithm \ref{ModifiedNewtonAlg} terminates after at most $\kappa$ steps.
\end{la}
\begin{proof}
Since $A_{s_0} + \rho\sum_{s\in\mathbb{S}\backslash\{s_0\}}A^{x^n}_s$ is a non-singular matrix, we may write $\eqref{ModifiedNewtonAlg_Eq2}$ as
\begin{align}
x^{n+1}
= \big(A_{s_0} + \rho\sum_{s\in\mathbb{S}\backslash\{s_0\}}A^{x^n}_s\big)^{-1}\big(b_{s_0} +\rho \sum_{s\in\mathbb{S}\backslash\{s_0\}} b^{x^n}_s\big).\label{ModifiedNewtonAlg_Eq3}
\end{align}
There is only a finite number of compositions for the right hand side of $\eqref{ModifiedNewtonAlg_Eq3}$
as it consists of given matrices of which some rows are set to zero. Hence, there exists a set $\mathcal{C}$, $|\mathcal{C}|<\infty$, $\mathcal{C}$ independent of $x^0$, such that $x^n\in\mathcal{C}$ for all $n\in\mathbb{N}$. Also, by the properties of the algorithm, if $x^{n'+1}=x^{n'}$
for some $n'\in\mathbb{N}$, then it is $x^n=x^{n'}$ for all $n\geq n'$. This, together with the monotonicity of the algorithm (cf.\,Lemma\,\ref{NewtonAlgMonotonicity}), proves the lemma.
\end{proof}

\begin{theorem}\label{NewtonLimitSolvesPenProblem}
Independently of the starting value $x^0$\,, Algorithm \ref{ModifiedNewtonAlg} converges (in finitely many steps) to a limit $x^*$ which solves Problem \ref{PenDiscreteProbDef}.
\end{theorem}
\begin{proof}
As the proof of convergence has already been given in Lemma \ref{NewtonAlg_FiniteTermination}, it only remains to show that the limit $x^*$ satisfies $G(x^*)=0$.
This follows easily from the fact that $x^*=x^{n+1}=x^n$ for all $n\geq\kappa$, which means that \eqref{ModifiedNewtonAlg_Eq1} reduces to $-G(x^*) = -G(x^n) = J_G(x^n)(x^{n+1}-x^n)=0$.
\end{proof}

\section{Examples of HJB Equations Arising in Finance}\label{ExHJBEqInFinance}

Let $r$, $q\geq 0$, $\sigma>0$ denote interest rate, (continuously paid) dividend rate and volatility, respectively. We introduce the Black-Scholes operator
\begin{equation}
\mathcal{L}^{r,q,\sigma}_{BS}V:=\frac{\partial V}{\partial t} + \frac{1}{2}\sigma^2S^2\frac{\partial^2 V}{\partial S^2} + (r-q)S\frac{\partial V}{\partial S} - rV,\label{BSOperator}
\end{equation}
where $V=V(S,t)$ is some function of $S$ and $t$.

\begin{remark}\label{FiniteDifferenceDiscr}
Let $M$, $N\in\mathbb{N}$. Suppose we want to solve equation \eqref{BSOperator} numerically backwards in time on the interval $[0,S_{max}]\times[0,T]$. Also, suppose that a terminal value and boundary conditions have been fixed.
Let time and space grids be given
by $\{jk : j+1\in\mathcal{M}\}$ and $\{ih : i+1\in\mathcal{N}\}$, respectively, where
$\mathcal{M}:=\{1,\ldots,M\}$, $\mathcal{N}:=\{1,\ldots,N\}$, $k := T/(M-1)$ and $h := S_{max}/(N-1)$.
Using a fully implicit finite difference method with a one-sided difference in time and central differences in space, we obtain the following scheme: For every time step $j\in\mathcal{M}\backslash\{1\}$, we have to find $V^{j-1}\in\mathbb{R}^N$ such that
\begin{equation*}
-A^{r,q,\sigma}_{BS} V^{j-1} + V^j = 0,
\end{equation*}
where $V^j\in\mathbb{R}^N$ is the vector known from the previous time step and $A^{r,q,\sigma}_{BS}\in\mathbb{R}^{N\times N}$ is a tridiagonal matrix with $(b_i)^N_{i=1}$ on the diagonal, $(c_i)^{N-1}_{i=1}$ on the upper diagonal, and $(a_i)^N_{i=2}$ on the lower diagonal; more precisely, for $i\in\mathcal{N}\backslash\{1,N\}$, the coefficients are given by
\begin{align*}
a^{r,q,\sigma}_i&\ = -\frac{1}{2}i^2\sigma^2k + \frac{1}{2}i(r-q)k,\\
b^{r,q,\sigma}_i&\ = 1 + i^2\sigma^2k + rk,\\
\text{and}\quad c^{r,q,\sigma}_i&\ = -\frac{1}{2}i^2\sigma^2k - \frac{1}{2}i(r-q)k,
\end{align*}
and $b^{r,q,\sigma}_1$\,, $c^{r,q,\sigma}_1$\,, $a^{r,q,\sigma}_N$ and $b^{r,q,\sigma}_N$ are chosen to accommodate the boundary conditions.
In finance, boundary conditions can usually be found such that $A^{r,q,\sigma}_{BS}$ is an M-matrix.
\end{remark}
\begin{proof}
See \cite{Seydel_ToolsCompFinance}.
\end{proof}

The examples from mathematical finance listed in Section \ref{ProblemFormulation} (uncertain volatility models \cite{Avellaneda_UncertainVol}, transaction cost models \cite{Leland_TransactionCosts},
and unequal borrowing/lending rates and stock borrowing fees \cite{Amadori_NonlinINtegroDiffProb_OptionPricing_ViscositySolApproach,Bergman_DifferentialINterestRates,Duffie_LendingShortingPricing}), all result
in HJB equations of the form \eqref{DifferentialOperatorProblem_Eq1} where every differential operator
$\mathcal{L}_s$\,, $s\in\mathbb{S}$, can be interpreted to be of Black-Scholes type \eqref{BSOperator} for some set of parameters $r$, $q$ and $\sigma$. (For details, we refer to \cite{Forsyth_Controlled_HJB_PDEs_Finance,Forsyth_NonlinearPDEFinance}, where the models have been described at length.) As seen in Remark \ref{FiniteDifferenceDiscr}, a finite difference discretisation of \eqref{BSOperator} results, for every timestep, in a linear system of equations where the defining matrix is an M-matrix, and, therefore, the above problems all fit into the framework of Problem \ref{DiscreteProbDef}; a detailed outline of one representative case will be given in the next section.

\subsection{Example: European Option Pricing with Stock Borrowing Fees and Unequal Borrowing/Lending Rates}\label{StableConvDiscretisations}
We would like to price a European option that pays $P(S)$ at time $T>0$, where $S$ denotes the price of the underlying. We assume the standard Black-Scholes model with the extension that the cash borrowing rate (denoted by $r_b$) and the cash lending rate (denoted by $r_l$) are not necessarily equal, and, also, that there are stock borrowing fees of $r_f$ reflecting the fact that shorting the stock may have a certain cost. To preclude arbitrage, it must be $r_b\geq r_l\geq r_f$\,.
Different borrowing/lending rates have been discussed in \cite{Amadori_NonlinINtegroDiffProb_OptionPricing_ViscositySolApproach,Bergman_DifferentialINterestRates}, and stock borrowing fees are studied in \cite{Duffie_LendingShortingPricing}; the resulting mathematical model has been explained in \cite{Forsyth_Controlled_HJB_PDEs_Finance}.\bigskip

To price a short position in the option, we have to solve the following nonlinear problem. We define $S_{max}>0$ to be a large upper limit on the value of the underlying $S$ and set $\Omega:=(0,S_{max})$.
For simplicity of presentation, we assume that $P(0)=P(S_{max})=0$; most other boundary conditions
can be treated similarly.

\begin{prob}\label{BorrowLendStockFeesProbDef}
Let $\sigma>0$ and $r_b\geq r_l\geq r_f\geq 0$\,.
Find $V:\overline{\Omega}\times [0,T]\to\mathbb{R}$ such that $V(S,t)=P(S)$ for $(S,t)\in\partial \Omega\times (0,T]$, $V(S,T) = P(S)$ for $S\in\overline{\Omega}$, and
\begin{equation}
\max\big\{
\mathcal{L}^{r,q,\sigma}_{BS}V : (r,q)\in\mathbb{S} 
\big\}= 0\label{BorrowLendStockFeesProbDef_Eq1}
\end{equation}
on $\Omega\times(0,T)$; here, it is $\mathbb{S}:=\{(r_l\,,0),(r_b\,,0),(r_l\,,r_f),(r_b\,,r_b-r_l+r_f)\}$.
\end{prob}

We use the notion of viscosity (sub-/super-) solution as developed in \cite{UsersGuide_ViscositySols}, and, in that framework, the following strong comparison principle -- which guarantees the uniqueness of any viscosity solution -- can be shown to hold for the pricing problem.

\begin{theorem}
Suppose $u$ and $v$ are, respectively, viscosity sub- and supersolution of Problem \ref{BorrowLendStockFeesProbDef}. If the payoff satisfies $P\in\mathcal{C}(\overline{\Omega})$, then it is
$u\leq v$ on $\overline{\Omega}\times (0,T]$.
\end{theorem}
\begin{proof}
After transforming the terminal value problem into an initial value problem, this can be found directly by checking the necessary conditions in \cite{UsersGuide_ViscositySols} or \cite{SebChaumont_StronCompPrinc}.
\end{proof}

It now seems natural to apply the discretisation technique from Remark \ref{FiniteDifferenceDiscr} to
every Black-Scholes operator in equation \eqref{BorrowLendStockFeesProbDef_Eq1} individually to obtain a nonlinear discrete approximation of Problem \ref{BorrowLendStockFeesProbDef}.

\begin{prob}\label{BorrowLendStockFeesProbDef_Discretised}
Let $M$, $N\in\mathbb{N}$ and use the notation of Remark \ref{FiniteDifferenceDiscr}. For $\sigma$ and $\mathbb{S}$ as in Problem \ref{BorrowLendStockFeesProbDef}, find $(V^j)_{j\in\mathcal{M}}\subset\mathbb{R}^N$ such that
\begin{equation}
\min\big\{A^{r,q,\sigma}_{BS}V^{j-1}-V^j : (r,q)\in\mathbb{S}\big\} = 0,\quad j\in\mathcal{M}\backslash\{1\},\label{BorrowLendStockFeesProbDef_Discretised_Eq1}
\end{equation}
where
\begin{itemize}
\item $V^M\in\mathbb{R}^N$ is the function $P$ evaluated on the grid $\{ih: i+1\in\mathcal{N}\}$
\item and $b^{r,q,\sigma}_1=b^{r,q,\sigma}_N=1$ and $c^{r,q,\sigma}_1=a^{r,q,\sigma}_N=0$ in
all matrices $A^{r,q,\sigma}_{BS}$.
\end{itemize}
\end{prob}

What follows are a few properties -- stability, consistency and monotonicity, used as defined in \cite{BarlesMainArticle} -- of the just introduced numerical scheme; based on these, we will then be able to deduce convergence.

\begin{la}\label{BorrowLendStockBorrow_StabilityLemma}
If $\sigma^2\geq r_b$\,, the discretisation in Problem \ref{BorrowLendStockFeesProbDef_Discretised} is unconditionally stable in the maximum norm.
\end{la}
\begin{proof}
For any $(r,q)\in\mathbb{S}$ and $i\in\mathcal{N}\backslash\{1,N\}$, it is $b^{r,q,\sigma}_i\geq 1$ and $a^{r,q,\sigma}_i$\,, $c^{r,q,\sigma}_i\leq 0$.
Hence, from \eqref{BorrowLendStockFeesProbDef_Discretised_Eq1}, for any given $j\in\mathcal{M}\backslash\{1\}$ and $i\in\mathcal{N}\backslash\{1,N\}$ there is choice of $(r,q)\in\mathbb{S}$ such that
\begin{align*}
|V^{j-1}_i| =&\ \frac{1}{b^{r,q,\sigma}_i}|V^j_i - a^{r,q,\sigma}_iV^{j-1}_{i-1} - c^{r,q,\sigma}_i V^{j-1}_{i+1}|\\
\leq &\ \frac{1 - a^{r,q,\sigma}_i - c^{r,q,\sigma}_i }{b^{r,q,\sigma}_i}\max\big(|V^j_i|,|V^{j-1}_{i-1}|,|V^{j-1}_{i+1}|\big),
\end{align*}
and we infer that
\begin{equation*}
\max_{1\leq i\leq N}|V^{j-1}_i|\leq \max\big(|V^{j-1}_1|,|V^{j-1}_N|,\max_{1\leq i\leq N}|V^j_i|\big)\leq \max_{1\leq i\leq N}P(ih),\quad j\in\mathcal{M}\backslash\{1\}.
\end{equation*}
\end{proof}

We point out that the constraint on the size of $\sigma$ required in Lemma \ref{BorrowLendStockBorrow_StabilityLemma} can be avoided when using upwinding instead of central differences for the space derivatives. However, since, in practice, the implicit scheme runs stably regardless of the choice of $\sigma$, we do not investigate this point any further.

\begin{la}\label{BorrowLendStockBorrow_ConsistencyLemma}
The discretisation in Problem \ref{BorrowLendStockFeesProbDef_Discretised} is consistent.
\end{la}
\begin{proof}
Let $\phi:\overline{\Omega}\times [0,T]\to\mathbb{R}$, and, for $i\in\mathcal{N}$ and $j\in\mathcal{M}$, write $\phi^j_i:=\phi(ih,jk)$.
We then define 
\begin{multline*}
\mathfrak{S}(h,k,ih,jk,\phi^j_i\,,\phi):=\\
\begin{cases} a^{r,q,\sigma}_i \phi^j_{i-1} + b^{r,q,\sigma}_i \phi^j_i + c^{r,q,\sigma}_i \phi^j_{i+1} - \phi^{j+1}_i & \text{if $i\in\mathcal{N}\backslash\{1,N\}$ $\wedge$ $j\in\mathcal{M}\backslash\{M\}$,}\\
\phi^j_i-P(ih) &\text{if $i\in\{1,N\}$ $\vee$ $j=M$,}
\end{cases}
\end{multline*}
and, for general $(S,t)\in\Omega\times(0,T)$ -- which is not necessarily on the grid -- we set
\begin{equation*}
\mathfrak{S}(h,k,S,t,\phi(S,t),\phi):=\,\mathfrak{S}(h,k,ih,jk,\phi^j_i\,,\phi) 
\end{equation*}
using $i = \lceil S/h\rceil h$ and $j = \lceil t/k\rceil k$.
This definition of the function $\mathfrak{S}$ is in line with the notation in \cite{BarlesMainArticle}, and, clearly, our numerical scheme from Remark \ref{FiniteDifferenceDiscr} corresponds to $\mathfrak{S}(h,k,ih,jk,V^j_i\,,(V^j)_{j\in\mathcal{M}})=0$.
Now, if $\phi$ is smooth on $\overline{\Omega}\times [0,T]$, we have
\begin{equation*}
\liminf
\frac{
\mathfrak{S}(h,k,S,t,\phi(S,t),\phi)}{k}\geq \zeta(S_0\,,t_0)
\end{equation*}
as $h\to 0$, $k\to 0$, $t\to t_0$ and $S\to S_0$\,, where
\begin{equation*}
\zeta(S_0\,,t_0) :=
\begin{cases}
(\mathcal{L}_{BS}^{r,q,\sigma}\phi)(S_0\,,t_0)
&\text{if $(S_0\,,t_0)\in\Omega\times (0,T)$,}\\
\min\big\{(\mathcal{L}_{BS}^{r,q,\sigma}\phi)(S_0\,,t_0)\,,\, \phi(S_0\,,t_0)-P(S_0)   \big\}
&\text{if $S_0\in\partial\Omega$ or $t_0=T$,}
\end{cases}
\end{equation*}
and the relation also holds when replacing ``$\liminf$'', ``$\geq$'' and ``$\min$'' by ``$\limsup$'', ``$\leq$'' and ``$\max$'', respectively.
This proves the lemma.
\end{proof}

\begin{la}\label{BorrowLendStockBorrow_MonotonicityLemma}
The discretisation in Problem \ref{BorrowLendStockFeesProbDef_Discretised} is monotone.
\end{la}
\begin{proof}
We use the notation of Lemma \ref{BorrowLendStockBorrow_ConsistencyLemma}.
It is
\begin{equation*}
\mathfrak{S}(h,k,S,t,z,\phi)\leq \mathfrak{S}(h,k,S,t,z,\psi)
\end{equation*}
if $\phi\geq\psi$ since $a^{r,q,\sigma}_i$, $c^{r,q,\sigma}_i\leq 0$.
\end{proof}

Finally, we have done enough preparatory work to formulate the following theorem on the convergence properties of our numerical scheme.

\begin{theorem}
In the limit $h$, $k\to 0$, the solution $(V^j)_{j\in\mathcal{M}}$ to Problem \ref{BorrowLendStockFeesProbDef_Discretised} converges uniformly on $\overline{\Omega}\times[0,T]$ to the unique viscosity solution to Problem \ref{BorrowLendStockFeesProbDef}.
\end{theorem}
\begin{proof}
Using Lemmas \ref{BorrowLendStockBorrow_StabilityLemma}, \ref{BorrowLendStockBorrow_ConsistencyLemma} and \ref{BorrowLendStockBorrow_MonotonicityLemma}, this follows from results in \cite{BarlesMainArticle}.
\end{proof}

\section{Numerical Results}\label{Section_NumResults}

In this section, we numerically solve Problem \ref{BorrowLendStockFeesProbDef} using the penalty approach presented in Section \ref{Section_PenalisationDiscrProb}. To have a non-trivial financial instrument, we choose
\begin{equation*}
P(S) = \begin{cases}
S/4-25 &\text{if $S\in[100,200)$},\\
-S/4+75 &\text{if $S\in[200,300)$},\\
0 &\text{else},
\end{cases}
\end{equation*}
which is the payoff function of a butterfly spread; unless stated differently, all other parameters are set as given in Table \ref{tab:DefaultParam}.\bigskip

\begin{table}[b]
\begin{tabular}{|c|c|c|c|c|c|c|c|c|c|}
\hline
$r_b$ & $r_l$ & $r_f$  & $\sigma$ & T & $S_{max}$ & M & N & tol & $\rho$\\ \hline
0.15 & 0.1 & 0.08 & 0.4 & 1 & 600 & 400 & 400 & 1e-08 & 1e04 \\ \hline
\end{tabular}
\caption{The parameters used in the numerical computations. (For illustrative purposes, we take unrealistically high interest rates.)}
\label{tab:DefaultParam}
\end{table}

After discretising the problem (cf.\,Section\,\ref{StableConvDiscretisations}), we are left with the nonlinear discrete equation in Problem \ref{BorrowLendStockFeesProbDef_Discretised} (which corresponds directly to Problem \ref{DiscreteProbDef}). To solve it, we employ (i) the penalty approach from Section \ref{Section_PenalisationDiscrProb} (which means we solve Problem \ref{PenDiscreteProbDef} instead) and (ii) the policy iteration as devised in \cite{Forsyth_Controlled_HJB_PDEs_Finance,Forsyth_NonlinearPDEFinance}; a detailed outline of policy iteration adapted to our case is given in Appendix \ref{Appendix_PolicyIteration}.
The following result is very similar to Lemma \ref{NewtonAlg_FiniteTermination}.

\begin{prop}\label{PolicyIt_FiniteTermination}
There exists a number $\lambda\in\mathbb{N}$ such that, for every starting value $x^0$\,, the method of policy iteration terminates after at most $\lambda$ steps.
\end{prop}
\begin{proof}
See Appendix \ref{Appendix_PolicyIteration}.
\end{proof}

\begin{remark}\label{WhenToTerminateRemark}
The methods used to solve Problems \ref{DiscreteProbDef} and \ref{PenDiscreteProbDef} are of iterative nature, which means we need a valid check for convergence.
\begin{itemize}
\item When solving Problem \ref{PenDiscreteProbDef} by Algorithm \ref{ModifiedNewtonAlg}, we abort the iteration as soon as we find an $x^n_{\rho}$ such that
\begin{equation*}
\frac{\|(A_{s_0}\,x^n_{\rho} - b_{s_0}) -\rho \sum_{s\in\mathbb{S}\backslash\{s_0\}}\max\{b_s-A_s\,x^n_{\rho}\,,0\}\|_{\infty}}{\|\max\{b_s : s\in\mathbb{S}\}\|_{\infty}}\leq tol.
\end{equation*}
\item When solving Problem \ref{DiscreteProbDef} by policy iteration, we abort the iteration as soon as we find an $x^n$ such that
\begin{equation*}
\frac{(A_sx^n-b_s)_l}{\|\max\{b_s : s\in\mathbb{S}\}\|_{\infty}}\geq -tol\quad\forall\ (l,s)\in\mathcal{N}\times\mathbb{S}
\end{equation*}
and, for every $i\in\mathcal{N}$, there exists an $r\in\mathbb{S}$ such that
\begin{equation*}
\frac{(A_rx^n-b_r)_i}{\|\max\{b_s : s\in\mathbb{S}\}\|_{\infty}}\leq tol.
\end{equation*}
\end{itemize}
\end{remark}
\medskip

In Figures \ref{fig:Butterfly_CallSide} and \ref{fig:Butterfly_PutSide}, we see the solution to Problem \ref{BorrowLendStockFeesProbDef} for different values of $r_b$\,, $r_l$ and $r_f$.
To understand the effects, it helps to see the butterfly spread as having a \textit{call side} (the left hand side) and a \textit{put side} (the right hand side). For a short position in a vanilla call option, the hedged portfolio is long the underlying and short the bank account. Conversely, for a short position in a vanilla put option, the hedged portfolio is short the underlying and long the bank account.
Therefore, the call side of the butterfly spread rises when the borrowing rate is increased (cf.\,Figure\,\ref{fig:Butterfly_CallSide}), and the put side rises when the borrowing rate is increased or stock borrowing fees are introduced (cf.\,Figures\,\ref{fig:Butterfly_CallSide},\,\ref{fig:Butterfly_PutSide}).\bigskip

\begin{figure}[p]
\centering
\includegraphics[width=8cm,height=8cm]{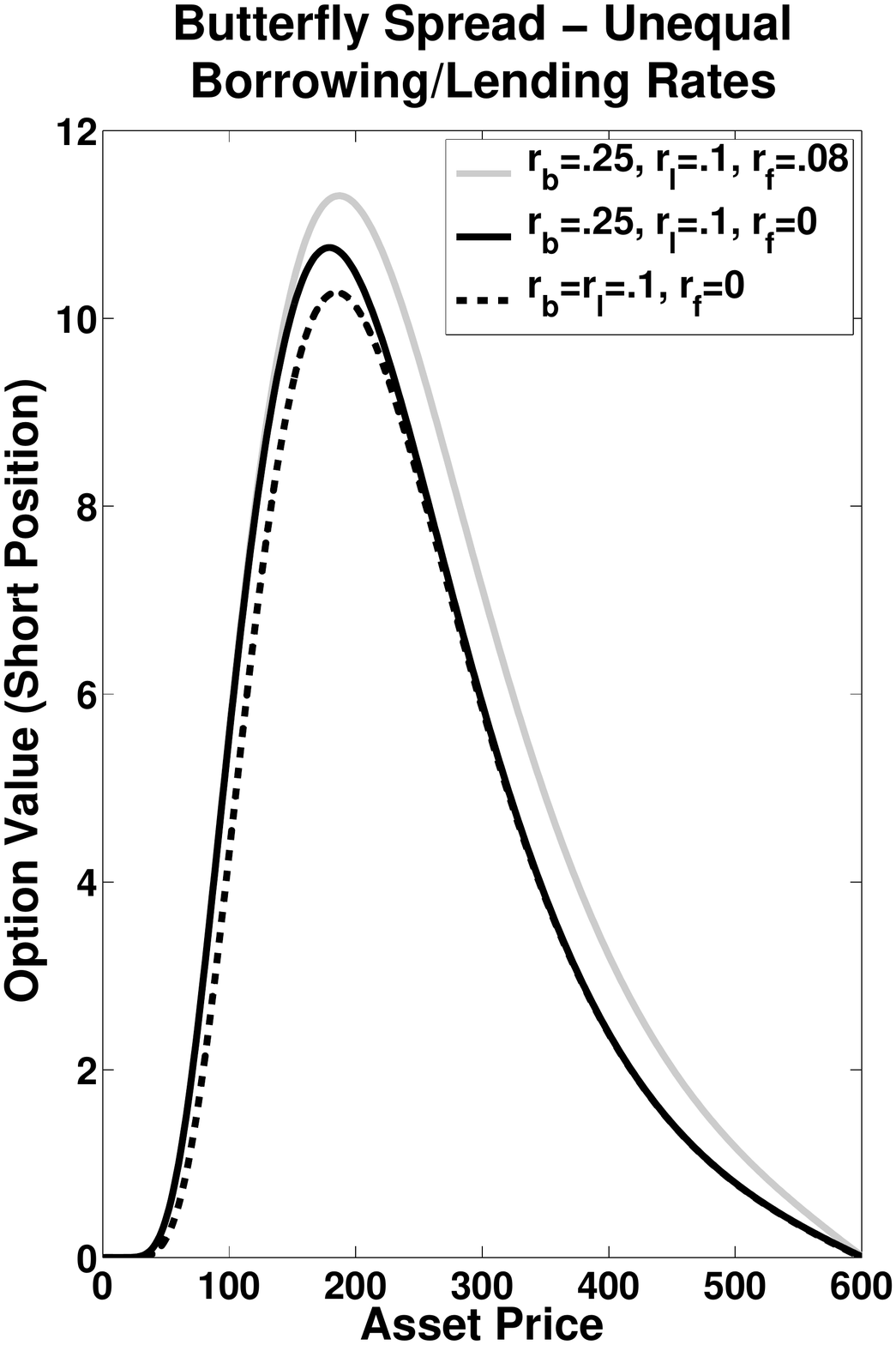}
\caption{The value of a short position in a butterfly spread. When raising the borrowing rate, the value on the left hand side (where a short position in the bank account is needed for hedging) increases and the right hand side does not change. When introducing a fee for stock borrowing, the right hand side (where a short position in the underlying is needed for hedging) increases and the left hand side does not change.}
\label{fig:Butterfly_CallSide}
\end{figure}

\begin{figure}[p]
\centering
\includegraphics[width=8cm,height=8cm]{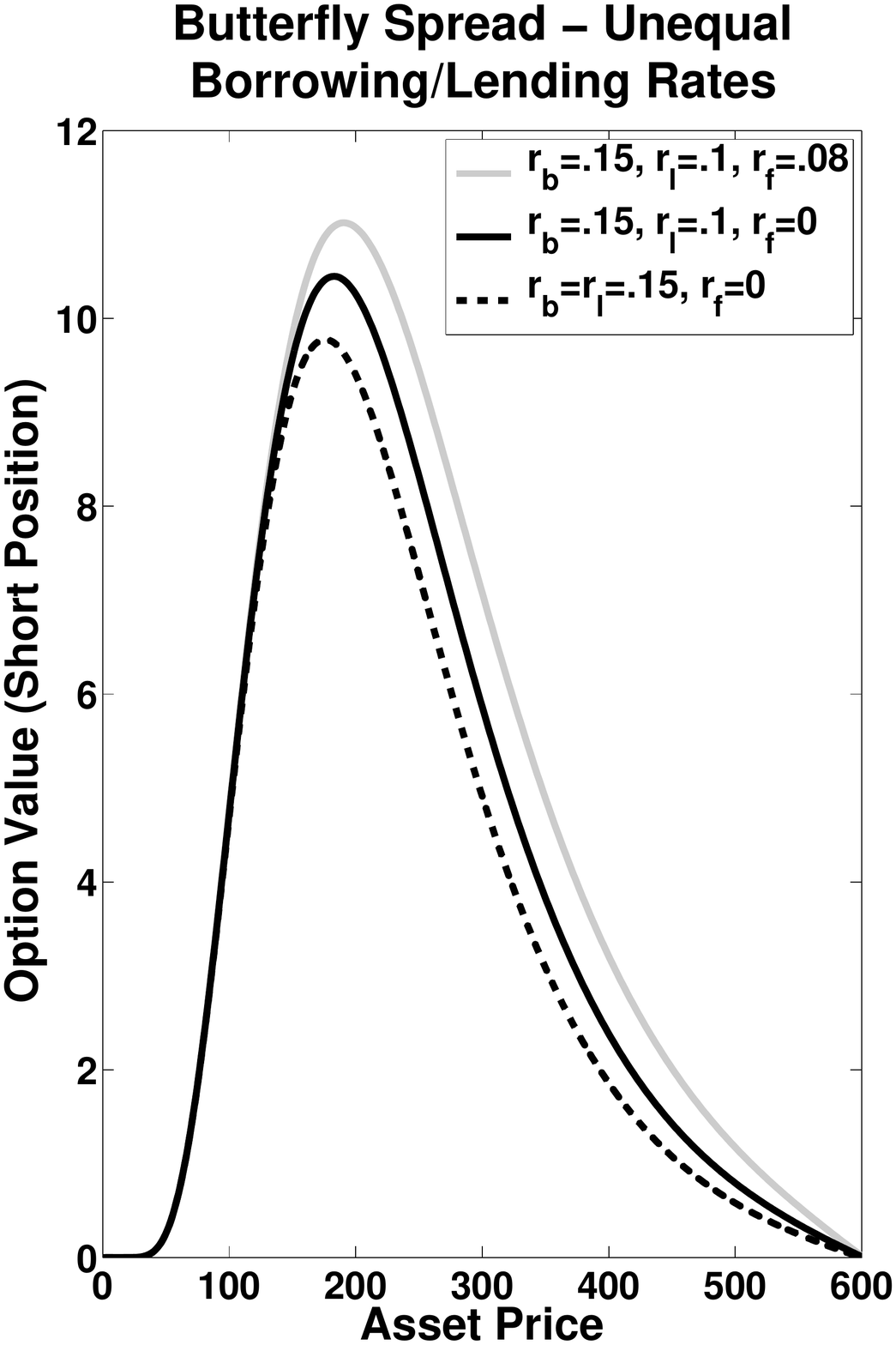}
\caption{The value of a short position in a butterfly spread. When decreasing the lending rate, the value on the right hand side (where a long position in the bank account is needed for hedging) increases and the left hand side does not change. When introducing a fee for stock borrowing, the right hand side (where a short position in the underlying is needed for hedging) increases again and the left hand side again does not change.}
\label{fig:Butterfly_PutSide}
\end{figure}

\begin{table}[t]
\begin{tabular}{|l|c|c|c|c|c|}
\hline
\textit{Policy Iteration} & $n=1$ & $n=2$ & $n=3$ & $n=4$\\
  \hline
$M$, $N=400$ & 90.50$\%$ & 9.50$\%$ & - & - \\
  \hline
$M$, $N=1000$ & 91.20$\%$ & 8.80$\%$ & - & - \\
\hline
$M=900$, $N=30$ & 99.67$\%$ & 0.33$\%$ & - & - \\
\hline
$M=30$, $N=900$ & - & 96.67$\%$ & 3.33$\%$ & - \\
\hline
\hline
\textit{Penalty Method} $(\rho = 4e03)$ & $n=1$ & $n=2$ & $n=3$ & $n=4$\\
  \hline
$M$, $N=400$ & - & - & 78.75$\%$ & 21.25$\%$ \\
  \hline
$M$, $N=1000$ & - & - & 78.40$\%$ & 21.60$\%$ \\
\hline
$M=900$, $N=30$ & - & - & 92.56$\%$ & 7.44$\%$ \\
\hline$M=30$, $N=900$ & - & - & 66.67$\%$ & 33.33$\%$ \\
\hline
\hline
\textit{Penalty Method} $(\rho = 1e06)$ & $n=1$ & $n=2$ & $n=3$ & $n=4$\\
  \hline
$M$, $N=400$ & - & - & 79.00$\%$ & 21.00$\%$ \\
  \hline
$M$, $N=1000$ & - & - & 78.20$\%$ & 21.80$\%$ \\
  \hline
$M=900$, $N=30$ & - & - & 92.44$\%$ & 7.56$\%$ \\
  \hline
$M=30$, $N=900$ & - & - & 70.00$\%$ & 30.00$\%$ \\
  \hline
\end{tabular}
\caption{The table shows the number of iterations (denoted by $n$) needed when using policy iteration and Algorithm \ref{ModifiedNewtonAlg} to solve the nonlinear discrete Problems \ref{DiscreteProbDef} and \ref{PenDiscreteProbDef}, respectively. Since the number of required iterations might differ between time steps, the percentages denote how frequently a certain number occurred. Indisputably, both schemes converge astonishingly well (since generally $n\leq4$), and, in many cases, policy iteration seems to even need only one step, whereas Algorithm \ref{ModifiedNewtonAlg} predominantly needs three.}
\label{tab:NoOfIteratios}
\end{table}

\begin{table}[t]
\begin{tabular}{|c|c|c|c|}
\hline
Grid Size & Policy & Penalty $(\rho = 4e03)$ & Penalty $(\rho = 1e06)$ \\
  \hline
$M$, $N=400$ & 0.216$s$ & 0.733$s$ & 0.740$s$ \\
  \hline
$M$, $N=1000$ & 1.148$s$ & 4.021$s$ & 4.093$s$ \\
\hline
$M=900$, $N=30$ & 0.142$s$ & 0.473$s$ & 0.475$s$ \\
\hline
$M=30$, $N=900$ & 0.0624$s$ & 0.116$s$ & 0.121$s$ \\
\hline
\end{tabular}
\caption{We see the computational time needed to solve Problems \ref{DiscreteProbDef} and \ref{PenDiscreteProbDef} using policy iteration and Algorithm \ref{ModifiedNewtonAlg}, respectively.
The penalty scheme needs roughly three to four times longer than the policy iteration, a difference which reflects the different numbers of iterations listed in Table \ref{tab:NoOfIteratios}.
}
\label{tab:RunTime}
\end{table}

To see how well the penalty approximation (cf.\,Section\,\ref{Section_PenalisationDiscrProb}) works in practice, we observe the difference between Problems \ref{DiscreteProbDef} and \ref{PenDiscreteProbDef} in dependence on the penalty parameter $\rho$, where we solve the former using policy iteration (see \cite{Forsyth_Controlled_HJB_PDEs_Finance,Forsyth_NonlinearPDEFinance} and Appendix\,\ref{Appendix_PolicyIteration}) and the latter using Algorithm \ref{ModifiedNewtonAlg}. We measure the difference in the maximum norm at time zero and keep the time and space grid fixed at $M$, $N=400$.
Clearly, Figure \ref{fig:Butterfly_PenaltyConvergence} shows that the observed convergence is of first order, which confirms the result of Corollary \ref{PenaltyConvergenceToTrueSol_ErrorEstimate}.
(We find almost identical convergence rates when using $(M,N)\in$ $\{(600,600),(1000,1000),(900,30),(30,900)\}$, i.e. the observed penalty convergence is independent of the grid size.)
\bigskip

\begin{figure}[t]
\centering
\includegraphics[width=8cm,height=8cm]{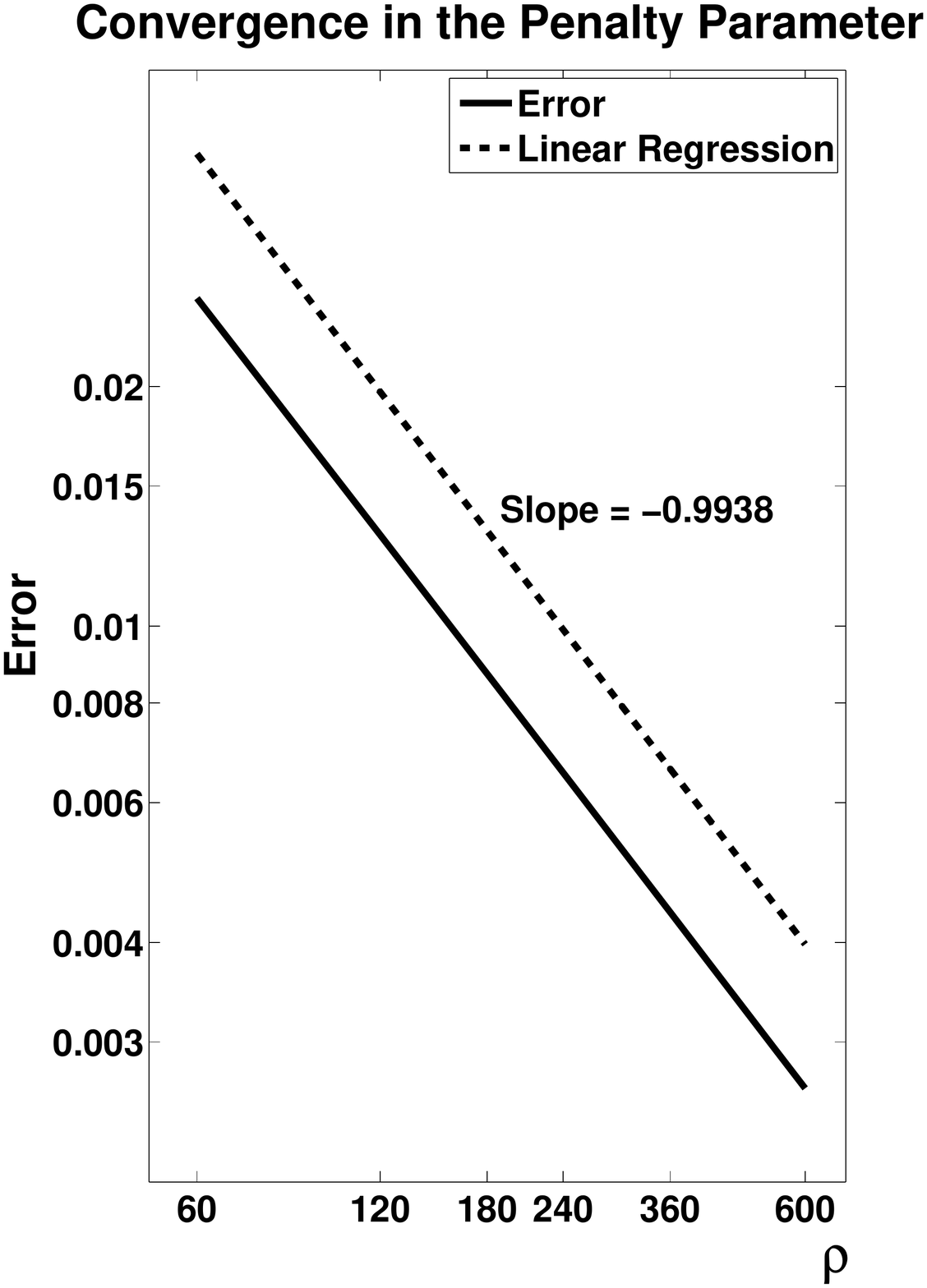}
\caption{We test the penalisation technique introduced in Section \ref{Section_PenalisationDiscrProb}. The plot (which is in log-log coordinates) shows the difference between the solutions to Problems \ref{DiscreteProbDef} and \ref{PenDiscreteProbDef} in dependence of the penalty parameter $\rho>0$. As $\rho$ is increased, the two problems converge nicely. Linear regression suggests the error is of order $O(1/\rho)$, confirming the analytical results (cf.\,Theorem\,\ref{PenaltyConvergenceToTrueSol_ErrorEstimate}).}
\label{fig:Butterfly_PenaltyConvergence}
\end{figure}

Finally, we come to analyse the numerical efficiency of the penalty method devised in this paper. Especially
in comparison with policy iteration, there are two relevant aspects to this: (i) the number of iterations needed to solve the nonlinear discrete system and (ii) the overall runtime of the method.\bigskip

The necessary numbers of iterations are listed in Table \ref{tab:NoOfIteratios}, and we see that, for the examples chosen, neither method ever needs more than four iterations; for the majority of cases, policy iteration needs only one iteration and Algorithm \ref{ModifiedNewtonAlg} needs merely three. The overall runtimes of both approaches are listed in Table \ref{tab:RunTime}, and, in consistency with the number of iterations needed for solving the nonlinear system, we see that the penalty method runs up to three to four times as long as the policy iteration; nicely, the runtime of the penalty scheme does not depend on the size of the penalty parameter.\bigskip

Since we generally expect finite termination (as stated in Lemma \ref{NewtonAlg_FiniteTermination} and Proposition \ref{PolicyIt_FiniteTermination}), one could, instead of proceeding as in Remark \ref{WhenToTerminateRemark}, also test for convergence by checking to what degree $x^n=x^{n+1}$ is satisfied. Interestingly, this is a disadvantage if a scheme converges in very few steps since one needs an extra step for verification once a certain accuracy is reached. (Using Remark \ref{WhenToTerminateRemark} to test for convergence, Table \ref{tab:NoOfIteratios} shows that, in most cases, the policy iteration reaches the desired accuracy in only one step. In \cite{Forsyth_Controlled_HJB_PDEs_Finance}, where it is tested for $|x^n-x^{n+1}|<tol$, the authors generally need two steps.)\bigskip

In total, for the current example, both approaches clearly work very well, with the policy iteration being computationally slightly more efficient (cf.\,Tables\,\ref{tab:NoOfIteratios},\,\ref{tab:RunTime}) than the penalty scheme.
Since, analytically, both algorithms have very similar properties, i.e. both algorithms converge monotonically (cf.\,Lemma \ref{NewtonAlgMonotonicity},\,\cite{Forsyth_Controlled_HJB_PDEs_Finance}) and in finitely many steps (cf.\,Lemma \ref{NewtonAlg_FiniteTermination}, Proposition \ref{PolicyIt_FiniteTermination}), it is hard to say generally which scheme should be faster.
Since the frameworks of this paper and \cite{Forsyth_Controlled_HJB_PDEs_Finance} are not identical\footnote{Loosely speaking, an HJB equation simply minimises over a (in our case finite) number of operators. The policy iteration presented in \cite{Forsyth_Controlled_HJB_PDEs_Finance} needs a time derivative to appear in all operators. The penalty scheme of this paper can easily be extended to include time independent operators, like the identity function used when comparing the value of the solution to the payoff in the case of American-style options.}
(with the example of Section \ref{Section_PenalisationDiscrProb} fitting into both setups), the question of which algorithm is advantageous might simply be down to the exact nature of the problem.

%
%


\section{Conclusion}

We have shown that penalty methods can be a powerful tool for the solution of HJB equations with finitely many controls, and we have included a rigorous analysis of the convergence properties of the overall approach. More precisely, we have discretised the HJB equation using standard finite differences and have approximated the resulting nonlinear discrete problem using a penalisation technique. If the differential operators of the equation  obey a certain structure, we have been able to prove stability and convergence of the finite difference approximation in the viscosity sense and to estimate the penalisation error to be of order $O(1/\rho)$, where $\rho>0$ is the penalty parameter. For the numerical solution of the penalised problem, we have devised an iterative method which has finite termination in theory and, for the considered examples, converges in less than four steps in practice. There is a wide range of problems in mathematical finance matching our framework; in particular, we have presented numerical results for European option pricing in a model with different borrowing/lending rates and stock borrowing fees and have demonstrated the competitiveness of our approach.

\appendix
\section{The Method of Policy Iteration}\label{Appendix_PolicyIteration}
We use this appendix to briefly describe the policy iteration scheme for the numerical solution of Problem \ref{BorrowLendStockFeesProbDef_Discretised} as devised in \cite{Forsyth_Controlled_HJB_PDEs_Finance}.
Also, we show how convergence and finite termination can easily be deduced.\bigskip

In the situation of Problem \ref{BorrowLendStockFeesProbDef_Discretised}, we define
\begin{equation*}
\mathfrak{A}_{BS}^{r,q,\sigma}:=\frac{1}{k}(A_{BS}^{r,q,\sigma}-I),\quad (r,q)\in\mathbb{S},
\end{equation*}
where $I\in\mathbb{R}^N$ denotes the identity matrix; in this new notation, Problem \ref{BorrowLendStockFeesProbDef_Discretised} can be reformulated: Find $(V^j)_{j\in\mathcal{M}}\subset\mathbb{R}^N$ such that
\begin{equation}
\frac{V^{j-1}-V^j}{k} +
\min\big\{\mathfrak{A}^{r,q,\sigma}_{BS}V^{j-1} : (r,q)\in\mathbb{S}\big\} = 0,\quad j\in\mathcal{M}\backslash\{1\}.\label{Appendix_Eq1}
\end{equation}
Given some $V^j$, $j\in\mathcal{M}\backslash\{1\}$, the following algorithm solves equation \eqref{Appendix_Eq1} for $V^{j-1}$.

\begin{alg}\label{PolicyIteration_Alg}(Policy Iteration)
Let $x^0\in\mathbb{R}^N$ be some starting value. Then, for known $x^n$, $n\geq 0$, find $x^{n+1}$ such that
\begin{equation}
(I+k\,\mathfrak{A}_{BS}^{r^n,q^n,\sigma})x^{n+1}=V^j,\label{PolicyIteration_Alg_Eq1}
\end{equation}
where $(r^n,q^n)\in\argmin\{\,\mathfrak{A}_{BS}^{r,q,\sigma}x^n:(r,q)\in\mathbb{S}\}$.
\end{alg}
Now, for $n\geq 1$, we have that
\begin{align*}
&(I+k\,\mathfrak{A}_{BS}^{r^n,q^n,\sigma})(x^{n+1}-x^n)\\
\geq\
&(I+k\,\mathfrak{A}_{BS}^{r^n,q^n,\sigma})x^{n+1}
-
(I+k\,\mathfrak{A}_{BS}^{r^{n-1},q^{n-1},\sigma})x^n=0,
\end{align*}
which, knowing that $I+k\,\mathfrak{A}_{BS}^{r^n,q^n,\sigma}$ is an M-matrix, means that $x^{n+1}\geq x^n$.
Furthermore, since $\{(I+k\,\mathfrak{A}_{BS}^{r,q,\sigma})^{-1}: (r,q)\in\mathbb{S}\}$ is a finite set, we may deduce the convergence of $(x^n)_{n\in\mathbb{N}}$ and the existence of a $\lambda\in\mathbb{N}$, independent of $x^0$, such that $x^{n+1}=x^{n}$ for all $n\geq \lambda$. If $x^{n+1}=x^n$, expression \eqref{PolicyIteration_Alg_Eq1} transforms into
\begin{align*}
(I+k\,\mathfrak{A}_{BS}^{r^n,q^n,\sigma})x^n=&\ V^j\\
\Leftrightarrow\quad
x^n +
k \min\big\{\,\mathfrak{A}^{r,q,\sigma}_{BS}x^n : (r,q)\in\mathbb{S}\big\} =&\ V^j,
\end{align*}
which means that $x^n$ solves \eqref{Appendix_Eq1}.

\renewcommand{\bibname}{References}

\bibliography{Paper_PenSchemeDiscrHJB_References}

\begin{thebibliography}{10}

\bibitem{Amadori_NonlinINtegroDiffProb_OptionPricing_ViscositySolApproach}
A.~L. Amadori.
\newblock Nonlinear integro-differential evolution problems arising in option
  pricing: A viscosity solution approach.
\newblock {\em Journal of Differential and Integral Equations}, 16(7):787--811,
  2003.

\bibitem{BarlesMainArticle}
G.~Barles.
\newblock Convergence of numerical schemes for degenerate parabolic equations
  arising in finance.
\newblock In L.~C. Rogers and D.~Talay, editors, {\em Numerical Methods in
  Finance}, pages 1--21. Cambridge: Cambridge University Press, 1997.

\bibitem{Barles_OntheConvergenceRate_ApproxHJBEq}
G.~Barles and E.~R. Jakobsen.
\newblock On the convergence rate of approximation schemes for
  {H}amilton-{J}acobi-{B}ellman equations.
\newblock {\em Mathematical Modelling and Numerical Analysis}, 36(1):33--54,
  2002.

\bibitem{Barles_ErrorBoundsMonotoneApproxSchemes}
G.~Barles and E.~R. Jakobsen.
\newblock Error bounds for monotone approximation schemes for parabolic
  {H}amilton-{J}acobi-{B}ellman equations.
\newblock {\em Mathematics of Computation}, 76(260):1861--1893, 2007.

\bibitem{Lions_ApplVarIneqStochControl}
A.~Bensoussan and J.~L. Lions.
\newblock {\em Applications of variational inequalities in stochastic control},
  volume~12 of {\em Studies in Mathematics and its Applications}.
\newblock Amsterdam, New York, Oxford: North-Holland Pub. Co.\hspace{2pt},
  1982.

\bibitem{Bergman_DifferentialINterestRates}
Y.~Z. Bergman.
\newblock Option pricing with differential interest rates.
\newblock {\em The Review of Financial Studies}, 8(2):475--500, 1995.

\bibitem{SebChaumont_StronCompPrinc}
S.~Chaumont.
\newblock Uniqueness to elliptic and parabolic {H}amilton-{J}acobi-{B}ellman
  equations with non-smooth boundary.
\newblock {\em Comptes Rendus Mathematique}, 339(8):555--560, 2004.

\bibitem{FiedlerSpecialMatrices}
M.~Fiedler.
\newblock {\em Special matrices and their applications in numerical
  mathematics}.
\newblock Lancaster: Nijhoff, 1986.

\bibitem{Fleming_Soner_ControlledMarkovProcesses}
W.~H. Fleming and H.~M. Soner.
\newblock {\em Controlled Markov processes and viscosity solutions}.
\newblock New York: Springer, 2nd edition, 2005.

\bibitem{Forsyth_Controlled_HJB_PDEs_Finance}
P.~A. Forsyth and G.~Labahn.
\newblock Numerical methods for controlled {H}amilton-{J}acobi-{B}ellman {PDE}s
  in finance.
\newblock {\em The Journal of Computational Finance}, 11(2):1--44, 2007.

\bibitem{ForsythQuadraticConvergence}
P.~A. Forsyth and K.~R. Vetzal.
\newblock Quadratic convergence for valuing {A}merican options using a penalty
  method.
\newblock {\em SIAM Journal on Scientific Computing}, 23(6):2095--2122, 2002.

\bibitem{Forsyth_NonlinearPDEFinance}
P.~A. Forsyth and K.~R. Vetzal.
\newblock Numerical methods for nonlinear {PDE}s in finance.
\newblock Working paper, University of Waterloo,
  http://www.rmi.nus.edu.sg/csf/webpages/Authors
  /seconddraft/forsyth-vetzal.pdf, 2010.

\bibitem{Duffie_LendingShortingPricing}
D.~Duffie{,}~N. G\^arleanu and L.~H. Pedersen.
\newblock Securities lending, shorting, and pricing.
\newblock {\em Journal of Financial Economics}, 66(2-3):307--339, 2002.

\bibitem{UsersGuide_ViscositySols}
M.~G. Crandall{,}~H. Ishii and P.-L. Lions.
\newblock User's guide to viscosity solutions of second order partial
  differential equations.
\newblock {\em Bulletin of the American Mathematical Society}, 27(1):1--67,
  1992.

\bibitem{ItoKunisch_SemiSmoothNewton_VarIneqFirstKind}
K.~Ito and K.~Kunisch.
\newblock Semi-smooth {N}ewton methods for variational inequalities of the
  first kind.
\newblock {\em Mathematical Modelling and Numerical Analysis}, 37(1):41--62,
  2003.

\bibitem{ItoKunisch_SemiSmoothNewtonAndGlobalization}
K.~Ito and K.~Kunisch.
\newblock On a semi-smooth {N}ewton method and its globalization.
\newblock {\em Mathematical Programming}, 118(2):347--370, 2007.

\bibitem{Kushner_NumMethodsStochControl_ContTime}
H.~J. Kushner.
\newblock Numerical methods for stochastic control problems in continuous time.
\newblock {\em SIAM Journal on Control and Optimization}, 28(5):99--1048, 1990.

\bibitem{KushnerDupuis_StochasticControlProblems_ContTime}
H.~J. Kushner and P.~G. Dupuis.
\newblock {\em Numerical methods for stochastic control problems in continuous
  time}.
\newblock New York: Springer, 2nd edition, 2001.

\bibitem{Leland_TransactionCosts}
H.~E. Leland.
\newblock Option pricing and replication with transactions costs.
\newblock {\em The Journal of Finance}, 40(5):1283--13--1, 1985.

\bibitem{Avellaneda_UncertainVol}
M.~Avellaneda{,}~A. Levy and A.~Par\'as.
\newblock Pricing and hedging derivative securities in markets with uncertain
  volatilities.
\newblock {\em Applied Mathematical Finance}, 2(2):73--888, 1995.

\bibitem{Santos_PolicyIteration}
M.~S. Santos and J.~Rust.
\newblock Convergence properties of policy iteration.
\newblock {\em SIAM Journal on Control and Optimization}, 42:2094--2115, 2004.

\bibitem{Seydel_ToolsCompFinance}
R.~Seydel.
\newblock {\em Tools for computational finance}.
\newblock Universitext. Berlin: Springer, 3rd edition, 2006.

\bibitem{Song_MarkovChainHJBEq}
Q.~S. Song.
\newblock Convergence of {M}arkov chain approximation on generalized {HJB}
  equation and its applications.
\newblock {\em Automatica}, 44(3):761--766, 2008.

\bibitem{Varga_IterativeMatrixAnalysis}
R.~S. Varga.
\newblock {\em Matrix iterative analysis}.
\newblock Berlin, London: Springer, 2000.

\bibitem{Wang_Forsyth_MaximalUSeCentralDifferences_HJBFinance}
J.~Wang and P.~Forsyth.
\newblock Maximal use of central differencing for {H}amilton-{J}acobi-{B}ellman
  {PDE}s in finance.
\newblock {\em SIAM Journal on Numerical Analysis}, 46(3):1580--1601, 2009.

\bibitem{XYZ_StochControlHJBBook}
J.~Yong and X.~Y. Zhou.
\newblock {\em Stochastic controls: Hamiltonian systems and HJB equations}.
\newblock New York, London: Springer, 1999.

\end{thebibliography}
\bibliographystyle{plain}

\end{document}